\documentclass[12pt]{article}

\usepackage{amsmath}
\usepackage{amsthm}
\usepackage{amssymb}
\usepackage{amsfonts}
\usepackage{graphicx}
\usepackage{mathptmx}

\newcommand{\E}{\operatorname{E}}
\newcommand{\V}{\operatorname{Var}}

\renewcommand{\(}{\left(}
\renewcommand{\)}{\right)}
\renewcommand{\bar}[1]{\overline{#1}}
\newcommand{\ce}{\mathit{ce}}
\newcommand{\re}{\mathit{re}}

\newcommand{\ds}{\displaystyle}

\newtheorem{thm}{Theorem}

\theoremstyle{definition}

\newtheorem{dfn}{Definition}

\numberwithin{equation}{section}

\begin{document}

\author{Helena Jasiulewicz\footnote{Wroc{\l}aw University of Environmental and Life Sciences,
Institute of Economics  and Social Sciences,
ul.~C.~K.~Norwida 25,
50-375 Wroc{\l}aw,
Poland, e-mail: helena.jasiulewicz@up.wroc.pl} \\
\setcounter{footnote}{6}
Wojciech Kordecki\footnote{University of Business in Wroc{\l}aw,
Department of Management,
ul.~A.~Ostrowskiego 22,
53-238 Wroc{\l}aw,
Poland, email: wojciech.kordecki@handlowa.eu}
}

\title{Ruin probability of a discrete-time risk process with proportional reinsurance and investment for exponential and Pareto distributions}

\date{}

\maketitle

\begin{abstract}
In this paper a quantitative analysis of the ruin probability in finite time of a discrete risk process with proportional reinsurance and investment of financial surplus is focused on. It is assumed that the total loss on a unit interval has a light-tailed distribution -- exponential distribution and a heavy-tailed distribution -- Pareto distribution. The ruin probability for finite-horizon 5 and 10 was determined from recurrence equations. Moreover, for exponential distribution the upper bound of ruin probability by Lundberg adjustment coefficient is given. For Pareto distribution the adjustment coefficient does not exist, hence an asymptotic approximation of the ruin probability if an initial capital tends to infinity is given. Obtained numerical results are given as tables and they are illustrated as graphs.
\end{abstract}


\section{Introduction}
\label{s:intro}

In the risk theory, works concerning the financial surplus of insurance companies in a continuous time have been proceeding for nearly a century. Very advanced models  of the classical continuous risk process were established. Although such a model is more natural in the description of reality, the research on the discrete process of financial surplus is considerably more modest. The review of the results concerning the discrete process of financial surplus one can find in the paper~\cite{HJ:surplus}. This paper is one of the series of papers which try to bring closer of the classical discrete process of financial surplus to the reality of insurance companies. Namely, the analysis of the investment of financial surplus enhances the security of an insurance company. These problems are considered in the
papers~\cite{Cai:Discrete,Cai:Ruin,CaiDickson:ruin,TangTsit_:Precise,Yang:Non-exp}.
Reinsurance has a considerable influence on increasing the security of an insurance company. The results concerning a discrete risk process with investment and reinsurance can be found in~\cite{DiasparaRomera:Bounds,HJ:ruin-discr}.

In this paper we consider the ruin probability in finite time of a discrete risk process with proportional reinsurance and investment of financial surplus.
Moreover, we obtain numerical results for particular cases: exponential and Pareto distribution of a total loss and some asymptotic results.

In the paper by Cai and Dickson~\cite{CaiDickson:ruin} the ruin probability in a discrete time risk process with a Markov interest model is studied. Recursive equations for the ruin probabilities, generalised Lundberg inequalities and an approximating approach to the recursive equations are given in that paper.
Diaspara and Romera~\cite{DiasparaRomera:Bounds} introduced a proportional reinsurance in the discrete risk process with an investment.

For any reinsurance, not only proportional, Jasiulewicz~\cite{HJ:ruin-discr} obtained recursive equations and Lundberg inequality for the ruin probability in the discrete-time risk process with Markovian chain interest rate model. Moreover,  for the proportional reinsurance and the reinsurance of stop-loss an optimal level of retention was considered, assuming the maximisation of Lundberg adjustment coefficient as an optimising criterion.

This paper is a continuation of the research initiated by
Jasiulewicz~\cite{HJ:ruin-discr}. For the given theoretical results we conduct a detailed quantitative analysis for particular distributions of the total loss in a unit period and proportional reinsurance. We consider the ruin probability for a light-tailed distribution (exponential pdf) and a heavy-tailed distribution (Pareto pdf) taking into account an investment of finance surplus according to a random interest rate.
Based on these considerations we give practical conclusions concerning connections between the initial capital level and the reinsurance level.
We pointed out the level of reinsurance of a loss in order to set a ruin probability
at the level low enough to be accepted by an insurer and vice versa i.e. how high his own capital should be.

The quality of the upper bound of ruin probability in finite time with the use of Lundberg coefficient  was illustrated by the example of exponential distribution. We observe that if an insurer and a reinsurer use the same security loading  then the  adjustment coefficient as a function of the reinsurance level is convex, which considerably improves an upper estimation of the ruin probability. However, if loading of a reinsurer is greater than loading of insurer, the adjustment coefficient is not a convex function, which lowers the quality of an upper estimation.
This observation was not taken into account in the numerical examples in Diaspara and Romera~\cite{DiasparaRomera:Bounds}.

It is known that for heavy-tailed distributions Lundberg adjustment coefficient does not exist. For distributions of that type we give the theorem about the approximation of the ruin probability if the initial capital is sufficiently large. The example of Pareto distribution shows that such an approximation is appropriate  and quickly tends to the limit value.

In the paper we assume the expectation of a loss in a unit period as a monetary unit. For that reason we assume that the expected values in both considered distributions are equal to 1. For the assumed values of parameters in Pareto distribution a variance does not exist. To compare numerical results for both distributions we also take such parameters in order to obtain the same geometric means as well as geometric variances.

Concluding, below we list the new elements, ideas and results  which are introduced in this article:
\begin{enumerate}
\item
In the continuous risk process the level of retention is optimal if it minimises the ruin probability which can be determined by maximising an adjustment coefficient relative to the level of retention
(see Dickson and Waters~\cite{DicksonWaters:Reinsurance}).
Then we can pose the following natural question: does the discrete risk process  hold the same?
\item
The upper bound of the ruin probability obtained by Lundberg coefficient in the case of proportional reinsurance is  given by
Diaspara and Romera~\cite{DiasparaRomera:Bounds}. The numerical example for $\xi=\theta$ shows that this estimation is reasonable. Is that estimation also good for the more natural case $\xi>\theta$?
\item
In the case of heavy tailed claims we give the approximation of the ruin probability.
The question is: is the sequence of approximations is fast convergent for sufficiently large initial capital?
\end{enumerate}

\section{Notations and theorems}
\label{s:not-thm}
Further notations, assumptions and theorems \ref{thm:ruin-recc1} and \ref{thm:est-up1} given below come from the paper by Jasiulewicz~\cite{HJ:ruin-discr}. In that paper the following notations and assumptions were taken.

\begin{enumerate}
\item\label{ass1}
Let $Z_n$ denote the total loss in unit period $\(n-1,n\right]$. The loss is calculated at the end of each period. Let us assume that $\left\{Z_n, n=1,2,\dots\right\}$ is a sequence of independent and identically distributed random variables with a common distribution function $W\(z\)$.
\item\label{ass2}

The premium is calculated by the expected value principle with the loading factor $\theta>0$. Constant premium $c=\(1+\theta\)\E Z_n$ is paid at the end of every unit period $\(n-1,n\right]$.
\item\label{ass3}
The insurer's surplus at the moment $n$ is denoted by $U_n$ and is calculated after the payoff. The surplus $U_n$ is invested at the beginning of the period $\(n,n+1\right]$ at a random rate $I_n$.
\item\label{ass4}
Let us assume that the interest rates $\left\{I_n, n=0,1,\dots\right\}$ follow a time-homogeneous Markov chain. We further assume that for all $n=0,1,\dots$, the rate $I_n$ takes possible values $i_1,i_2,\dots,i_l$. For all $n$ and all states, the transition probability is denoted by
\[
\Pr\(I_{n+1}=i_t|I_n=i_s\)=p_{st}\geq 0
\]
and the initial distribution is denoted by
\[
\Pr\(I_0=i_s\)=\pi_s\,.
\]
\item\label{ass5}
Suppose that the insurer effects reinsurance and that the amount paid by the insurer when the loss $Z_n$ occurs
is $h\(Z_n,b\)$ where a parameter $b>0$ denotes a retention level. The meaning of the parameter $b$ will be explained in two examples
of the most frequent reinsurancies applied in the insurance practice.
\begin{enumerate}
\item
Proportional reinsurance, if a function $h\(x,b\)$ has the form
\[
h\(x,b\)=bx,
\]
where $b\in\left(0,1\right]$.
\item
Stop loss reinsurance, if a function $h\(x,b\)$ has the form
\[
h\(x,b\)=
\begin{cases}
x, & x\leq b, \\
b, & x>b,
\end{cases}
\]
where $b>0$.
\end{enumerate}

The following assumption $0\leq h\(x,b\)\leq x$ about $h$ is obvious. A part of the loss $Z_n$ retained by the insurer is denoted by $Z_n^{\ce}=h\(Z_n,b\)$ and its distribution function by $V\(z\)$. Therefore $Z_n^{\re}=Z_n-Z_n^{\ce}$ is a reinsured part of the loss $Z_n$.

\item\label{ass6}
Le us assume that a reinsurer calculates a premium rate $c_{re}$ according to the expected value rule with a loading factor $\eta$, i.e.
\[
c_{re}=\(1+\eta\)\E\(Z_n-h\(Z_n,b\)\).
\]
We assume that $\eta\geq\theta>0$, so an insurer does not earn without risk if he retains only zero value of claims.
\item\label{ass7}
The premium rate retained by an insurer in a unit period is denoted by $c\(b\)$ and is given by
\[
c\(b\)
=c-c_{\re}
=\(1+\eta\)\E h\(Z_n,b\)-\(\eta-\theta\)\mu.
\]
\item\label{ass8}
Let $U_n^b$ denote a financial surplus of an insurer at the end of the unit period $\(n-1,n\right]$ after the payment of premium and after the payoff. The process $U_n^b$ considered in the paper is given by
\[
U_n^b=U_{n-1}^b\(1+I_n\)+c\(b\)-h\(Z_n,b\).
\]
\item\label{ass9}
The ultimate ruin probability for this risk process in the finite time is denoted by $\Psi_n^b\(u,i_s\)$ and is defined by
\[
\begin{split}
\Psi_n^b\(u,i_s\)&=\Pr\(\bigcup_{i=1}^n\(U_i^b<0\)|U_0^b=u,I_0=i_s\) \\
&=\Pr\(U_i^b<0\text{ for some $i\leq n$}|U_0^b=u,I_0=i_s\).
\end{split}
\]
The ultimate ruin probability in the infinite time is given by
\[
\begin{split}
\Psi^b\(u,i_s\)&=\Pr\(\bigcup_{i=1}^{\infty}\(U_i^b<0\)|U_0^b=u,I_0=i_s\) \\
&=\Pr\(U_i^b<0\text{ for some $i\geq 1$}|U_0^b=u,I_0=i_s\).
\end{split}
\]
Obviously
\[
\Psi^b\(u,i_s\)=\lim_{n\to\infty}\Psi_n^b\(u,i_s\).
\]
\end{enumerate}

The further research is conducted for a proportional reinsurance. The premium rate retained by an insurer is
\begin{equation}\label{eq:premium}
c\(b\)=\(\(1+\eta\)b-\(\eta-\theta\)\)\mu.
\end{equation}
To avoid such an event that the ruin could occur with probability 1 it is assumed that
\begin{equation}\label{eq:premium1}
\E h\(Z_1,b\)<c\(b\).
\end{equation}

To write the self-contained paper, we give theorems from Jasiulewicz~\cite{HJ:ruin-discr} (Theorems~\ref{thm:ruin-recc1} and \ref{thm:est-up1}), which will be used in the analysis of the ruin probability. In the special case of reinsurance, namely proportional reinsurance,
the theorems analogous to Theorems~\ref{thm:ruin-recc1} and \ref{thm:est-up1} were given
in the paper by Diaspara and Romera~\cite{DiasparaRomera:Bounds}.

\begin{thm}\label{thm:ruin-recc1}
Ruin probability of an insurer in finite time is given recursively in the following way:
\begin{align}
\label{eq:ruin-recc11}
\Psi_1^b\(u,i_s\)&=\sum_{j=1}^l p_{sj}\bar{V}\(u\(1+i_j\)+c\(b\)\), \\
\label{eq:ruin-recc12}
\begin{split}
\Psi_{n+1}^b\(u,i_s\)&=\sum_{j=1}^l p_{sj}\Big\{\bar{V}\(u\(1+i_j\)+c\(b\)\) \\
&+\int\limits_0^{u\(1+i_j\)+c\(b\)}\Psi_n^b\(u\(1+i_j\)+c\(b\)-z,i_j\)\,dV\(z\)\Big\}.
\end{split}
\end{align}
Ruin probability in infinite time:
\[
\begin{split}
\Psi^b\(u,i_s\)&=\sum_{j=1}^l p_{sj}\Big\{\bar{V}\(u\(1+i_j\)+c\(b\)\) \\
&+\int\limits_0^{u\(1+i_j\)+c\(b\)}\Psi^b\(u\(1+i_j\)+c\(b\)-z,i_j\)\,dV\(x\)\Big\},
\end{split}
\]
where
\begin{equation}\label{eq:c(b)}
c\(b\)=\(1+\eta\)\E h\(Z_n,b\)-\(\eta-\theta\)\mu.
\end{equation}
\end{thm}

\begin{proof}
Let $Z_1^{\ce}=z$, $I_1=i_j$. If $z>u\(1+i_j\)+c\(b\)$, then a ruin will occur in the first period $\(0,1\right]$. Therefore
\[
\begin{split}
\Psi_1^b\(u,i_s\)&=\sum_{j=1}^l p_{sj}\Pr\(Z_1^b>u\(1+i_j\)+c\(b\)|I_1=i_1,I_0=i_s\) \\
&=\sum_{j=1}^l p_{sj}\bar{V}\(u\(1+i_j\)+c\(b\)\).
\end{split}
\]
The ruin in first $n+1$ periods can occur in two excluding ways:
\begin{itemize}
\item
the ruin will occur in the first period or
\item
the ruin will not occur in the first period but it will occur in next periods.
\end{itemize}
Since the process $U_n^b$ is stationary with independent increments then
\[
\begin{split}
\Psi_{n+1}^b\(u,i_s\)
&=\sum_{j=1}^l p_{sj}\int\limits_0^{\infty}\Pr\(\bigcup_{k=1}^{n+1}
\(u_k^b<0|Z_1^b=z,I_1=i_s\)\)dV\(z\) \\
&=\sum_{j=1}^l p_{sj}\Bigg(\bar{V}\(u\(1+i_j\)+c\(b\)\) \\
&+\int\limits_0^{u\(1+i_j\)+c\(b\)}
\Psi_n^b\(u\(1+i_j\)+c\(b\)-z,i_j\)dV\(z\)\Bigg).
\end{split}
\]
The probability of the ruin in infinite time is obtained by taking a two-sided limit in the above formula for $n\to\infty$.
\end{proof}

Recurrence formulas for the  ruin probability can be presented in a matrix form, which simplifies calculations using several
computer programs\footnote{In this paper the calculations were made by program \textit{Maxima}:  {http://maxima.sourceforge.net/ .}{}
}.

Let
\[
{\pmb\Psi}_n^b\left(u\right)
=\left[\Psi_n^b\left(u,i_1\right),\Psi_n^b\left(u,i_2\right),\dots,
\Psi_n^b\left(u,i_l\right)\right]
\]
and
\[
\vec{V}_n=\left[v^{\left(n\right)}_1,v^{\left(n\right)}_2,\dots,
v^{\left(n\right)}_l\right],
\]
where
\begin{align*}
v^{\left(1\right)}_j&=\bar{V}\left(u\left(1+i_j\right)+c\left(b\right)\right) \\
\intertext{and for $n\geq 2$}
v^{\left(n+1\right)}_j&=v^{\left(1\right)}_j
+\int\limits_0^{u\left(1+i_j\right)+c\left(b\right)}
\Psi_n^b\left(u\left(1+i_j\right)
+c\left(b\right)-z,i_j,\right)\,dV\left(z\right).
\end{align*}
Then we can write equations \eqref{eq:ruin-recc11} and \eqref{eq:ruin-recc12} in a matrix form
\[
{\pmb\Psi}_{n}^b\left(u\right)=\vec{V}_nP^T.
\]

\begin{thm}\label{thm:est-up1}
If\, $\E h\(Z_1,b\)<c\(b\)$ and there exists a positive constant $R\(b\)$ fulfilling the equation
\begin{equation}\label{eq:coeff_R}
\E e^{R\(b\)h\(Z_1,b\)}=e^{R\(b\)c\(b\)},
\end{equation}
the upper estimation of the ruin probability in finite and infinite time is in the form
\begin{equation}\label{eq:est-up1}
\Psi_n^b\(u,i_s\)\leq
\Psi^b\(u,i_s\)\leq
\xi\(b\) \E\(e^{-R\(b\)u\(1+I_1\)}|I_0=i_s\),
\end{equation}
where
\begin{equation}\label{eq:sup-beta}
\xi\(b\)=\sup_{x\geq c\(b\)}\frac{e^{R\(b\)x}\bar{V}\(x\)}
{\int\limits_x^{\infty}e^{R\(b\)z}dV\(z\)}\,,
\quad
0<\xi\(b\)\leq 1.
\end{equation}
\end{thm}
\begin{proof}
For every $x\geq 0$ we have
\begin{equation}\label{eq:wsp-beta1}
\begin{split}
\bar{V}\(x+c\(b\)\)
&=\frac{e^{R\(b\)x}\bar{V}\(x+c\(b\)\)}
{\displaystyle\int\limits_x^{\infty}e^{R\(b\)z}\,dV\(z+c\(b\)\)}
e^{-R\(b\)x}\int\limits_x^{\infty}e^{R\(b\)z}\,dV\(z+c\(b\)\) \\
&=\frac{e^{R\(b\)\(x+c\(b\)\)}\bar{V}\(x+c\(b\)\)}
{\displaystyle \int\limits_{x+c\(b\)}^{\infty}e^{R\(b\)y}\,dV\(y\)}
e^{-R\(b\)x}\int\limits_{x+c\(b\)}^{\infty}e^{R\(b\)\(y-c\(b\)\)}\,dV\(y\).
\end{split}
\end{equation}
Let
\[
g\(t\)=
\frac{e^{R\(b\)\(t\)}\bar{V}\(t\)}
{\displaystyle \int\limits_{t}^{\infty}e^{R\(b\)y}\,dV\(y\)}.
\]
Then
\begin{equation}\label{eq:wsp-beta2}
\begin{split}
\bar{V}\(x+c\(b\)\)
&\leq\sup_{x\geq 0}
\left\{g\(x+c\(b\)\)\right\}
e^{-R\(b\)x}\int\limits_{x+c\(b\)}^{\infty}e^{R\(b\)\(y-c\(b\)\)}\,dV\(y\) \\
&=\beta
e^{-R\(b\)x}\int\limits_{x+c\(b\)}^{\infty}e^{R\(b\)\(y-c\(b\)\)}\,dV\(y\),
\end{split}
\end{equation}
where
\[
\beta=\sup_{y\geq c\(b\)}
g\(y\).
\]
From Equation~\eqref{eq:coeff_R} we obtain
\begin{equation}\label{eq:wsp-beta3}
\bar{V}\(x+c\(b\)\)
\leq
\beta
e^{-R\(b\)x}\int\limits_{-\infty}^{\infty}e^{R\(b\)\(y-c\(b\)\)}\,dV\(y\)
=\beta e^{-R\(b\)x}.
\end{equation}

Whereas the inequality~\eqref{eq:sup-beta} follows from the fact that for $z\geq t$ an inequality $\exp\(R\(b\)z\)\geq\exp\(R\(b\)t\)$ occurs. Therefore
\[
\frac{\displaystyle\int\limits_t^{\infty} e^{R\(b\)z}\,dV\(z\)}{e^{R\(b\)t}\bar{V}\(t\)}
\geq
\frac{\displaystyle e^{R\(b\)t}\int\limits_t^{\infty}dV\(z\)}{e^{R\(b\)t}\bar{V}\(t\)}=1.
\]

From the conversion of this inequality the inequality~\eqref{eq:sup-beta} is obtained.

In the next step we prove  \eqref{eq:est-up1} inductively.
From Theorem~\ref{thm:ruin-recc1} and inequality \eqref{eq:wsp-beta3} we have
\[
\Psi_1^b\(u,i_s\)\leq\sum_{j=1}^l p_{sj}\beta e^{-R\(b\)u\(1+i_j\)}
=\beta\E\(e^{-R\(b\)u\(1+I_1\)}|I_0=i_s\).
\]
From a inductive assumption
\[
\Psi_n^b\(u,i_s\)\leq
\beta\E\(e^{-R\(b\)u\(1+I_1\)}|I_0=i_s\)
\]
and Theorem~\ref{thm:ruin-recc1} we have
\[
\begin{split}
\Psi_{n+1}^b\(u,i_s\)
&\leq\sum_{i=j}^l p_{sj}
\Bigg(
\beta
e^{-R\(b\)u\(1+i_j\)}\int\limits_{u\(1+i_j\)+c\(b\)}^{\infty}e^{R\(b\)\(y-c\(b\)\)}\,dV\(y\) \\
&+\int\limits_0^{u\(1+i_j\)+c\(b\)}
\beta\E\(e^{-R\(b\)\(u\(1+i_j\)-z+c\(b\)\)\(1+I_1\)}|I_0=i_s\)
\Bigg).
\end{split}
\]
Since
\begin{equation}\label{eq:wsp-beta4}
\E\(e^{-R\(b\)\(u\(1+i_j\)-z+c\(b\)\)\(1+I_1\)}|I_0=i_s\)
\leq
e^{-R\(b\)\(u\(1+i_j\)-z+c\(b\)\)},
\end{equation}
then
\[
\begin{split}
\Psi_{n+1}^b\(u,i_s\)
&\leq \sum_{i=j}^l p_{sj}
\beta
e^{-R\(b\)u\(1+i_j\)}\int\limits_{-\infty}^{\infty}e^{R\(b\)\(y-c\(b\)\)}\,dV\(y\) \\
&=\beta\E\(e^{-R\(b\)u\(1+I_1\)}|I_0=i_s\).
\end{split}
\]
Taking limits for $n\to\infty$ we obtain the inequality~\eqref{eq:est-up1}.
\end{proof}

Theorem \ref{thm:ruin-recc1} gives recurrence formulae for the ruin probability and Theorem~\ref{thm:est-up1} gives an upper estimation of the ruin probability using Lundberg adjustment coefficient, which exists only for a light-tailed distribution. Therefore one cannot use Theorem~\ref{thm:est-up1} to estimate the ruin probability for heavy-tailed distributions. In that case we will use an asymptotic ruin probability in the respect of an initial capital tending to infinity, whereas the total loss has the distribution with a regularly varying tail.

\begin{dfn}
A distribution $F$ on $\(-\infty,\infty\)$ has a regularly varying tail if there exists some constant $\alpha\geq 0$ such that for every $y>0$ is
\[
\lim_{x\to\infty}\frac{\bar{F}\(xy\)}{\bar{F}\(x\)}=y^{-\alpha}.
\]
The class of such distributions is denoted by $\mathcal{R}_{-\alpha}$.
\end{dfn}

\begin{thm}\label{thm:est-up2}
Let total loss $Z_n$ have cdf\ $W\in\mathcal{R}_{-\alpha}$ for some $\alpha>0$. If\ $1+I_n>0$ for any fixed $I_0=i_s$ there exists a finite positive moment of rank $\alpha$ of discounting factor $\(1+I_1\)^{-1}$, then for a proportional reinsurance for every $I_0=i_s$ and every $n$ we have
\begin{equation}\label{eq:est-up2}
\Psi_n^b\(u,i_s\)\sim c_n\(i_s\)\bar{V}\(u\),
\end{equation}
if $u\to\infty$, where $c_n\(i_s\)$ are given recursively
\begin{equation}
\label{eq:cn}
c_n\(i_s\)=\E\(\(1+c_{n-1}\(I_1\)\)\(\frac{1}{1+I_1}\)^{\alpha}\bigg|I_0=i_s\),
\end{equation}
with an initial condition $c_0\(i_s\)=0$
for $n=1,2,\dots$
\end{thm}
\begin{proof}
In the paper by Cai and Dickson~\cite{CaiDickson:ruin}  the above theorem was proved in the case where an insurer does not apply reinsurance but invests the financial surplus.
It is sufficient to remark that with proportional reinsurance $Z_n^{\ce}=bZ_n$, if  $Z_n$ has a distribution with a regularly changing  tail with an index $\alpha$,
then $Z_n^{\ce}$ has also the distribution with a regularly varying  tail with an index $\alpha$. This follows from
\[
\lim_{x\to\infty}\frac{\bar{V}\(xy\)}{\bar{V}\(x\)}=
\lim_{x\to\infty}\frac{\bar{W}\(yx/b\)}{\bar{W}\(x/b\)}=
\lim_{z\to\infty}\frac{\bar{W}\(yz\)}{\bar{W}\(z\)}=y^{-\alpha},
\]
where $z=x/b\to\infty$, if $x\to\infty$, because $b>0$.
Therefore our Theorem~\ref{thm:est-up2} is fulfilled for $Z_{\ce}$ by Theorem 5.1 from the paper Cai and Dickson~\cite{CaiDickson:ruin}.
Our proof repeats the arguments given in Theorem 5.1 from that paper if we substitute $V$ with~$G$.
\end{proof}

In the next sections we will consider particular cases if the total loss in the unit period has an exponential distribution with mean 1, i.e. $W\(x\)=1-e^{-x}$ and has Pareto distribution with the same mean: $W\(x\)=1-\(\beta/x\)^{\alpha}$, $x>\beta$, $\alpha>1$, $\beta=\(\alpha-1\)/\alpha$. In Section~\ref{s:propo} we give analytical formulae only for the cases $l=1$, $i_1=0$
(i.e. financial surplus is not invested)
and small values of the parameter $n$.
To determine these formulae we use the program \textit{Maxima} assigned to symbolic calculations.

Numerical results will be presented for the case $l=2$ and for selected values of the parameters $\alpha$, $\beta$, $\eta$, $\theta$ and $b$.

\section{Ruin probability}
\label{s:propo}

Calculations of values of function $\Psi^b\(u,i_s\)$ given by Theorem~\ref{thm:ruin-recc1} were conducted for $b=0.2,\,0.3,\dots,1.0$, $u=0,\,1,\,2,\,3\,,4,\,5$ and $n=1,\,2,\dots,10$.
We considered the cases
\begin{itemize}
\item
$l=1$ for $i_1=0$,
\item
$l=2$ for $i_1=0.3$, $i_2=0.5$ with transition matrix
\[
P=
\begin{bmatrix}
0.4 & 0.6 \\
0.3 & 0.7
\end{bmatrix}.
\]
\end{itemize}
The values $\eta=0.25$ and $\theta=0.2$ were taken.
For $\E h\(Z_n,b\)=b$ from \eqref{eq:c(b)} we obtain  the formula
\[
c\(b\)=\(1+\eta\)b-\(\eta-\theta\)=1.25b-0.05.
\]
The condition \eqref{eq:premium1} is fulfilled for $b>1-\theta/\eta=0.2$.

\subsection{Exponential distribution}
\label{ss:expo}

Let us assume that $Z_n$ has the exponential distribution with mean 1. Hence $Z_n^{\ce}=b Z_n$ has the distribution function
\begin{equation}\label{eq:expo_pdf}
V\(x\)=1-e^{-x/b}
\end{equation}
for $x\geq 0$ and $\E Z_n^{\ce}=b$, $\V Z_n^{\ce}=b^2$.

The explicit formulae for function $\Psi_{1n}^b\(u,i_s\)$ for $n\geq 2$ are too complicated to present.
We take $l=1$ and $i_1=0$.
\begin{align}
\Psi_1^b\(u\)=&
{e}^{\frac{-u-\theta+\left( -b\right) \left( \eta+1\right) +\eta}{b}} \\
\Psi_2^b\(u\)=&
\begin{array}{l}
\ds\frac{\left( {e}^{2\eta/b}u+{e}^{2\eta/b}\theta
+\left( \left( b-1\right) \eta+b\right) {e}^{2\eta/b}\right) {e}^{-u/b-2\theta/b-2\eta-2}}{b} \\
+{e}^{\(-u-\theta-b\left( \eta+1\right) +\eta\)/b}
\end{array}
\end{align}
Formulae for $\Psi_n^b\(u\)$ for $n\leq 5$ obtained from \textit{Maxima} were used to verify the correctness of numerical algorithms which are used for greater $n$ and~$l$.

{
\tabcolsep 4pt
\begin{table}[!hbt]
\centering
\caption{\label{tab:ruin_exp}Values of ruin probabilities for exponential distribution}

\smallskip

\footnotesize

\begin{tabular}{|c|c|c|*{9}{c}|}
\hline
$n$ & $i_s$ & $u$ & \multicolumn{9}{c|}{$b$} \\\cline{4-12} 
& & &  0.2 & 0.3 & 0.4 & 0.5 & 0.6 & 0.7 & 0.8 & 0.9 & 1.0 \\\hline
5  & $3\%$ & 1 & 
0.0087 &   0.0385 &   0.0776 &   0.1164 &   0.1512 &   0.1814 &   0.2073 &   0.2299 &   0.2494 \\
&& 2 & 0.0001 &   0.0029 &   0.0119 &   0.0271 &   0.0460 &   0.0665 &   0.0871 &   0.1074 &   0.1265 \\  
&& 3 & 0.0000 &   0.0002 &   0.0017 &   0.0060 &   0.0134 &   0.0236 &   0.0357 &   0.0491 &   0.0630 \\   
&& 4 &   0.0000 &   0.0000 &   0.0002 &   0.0013 &   0.0038 &   0.0081 &   0.0143 &   0.0220 &   0.0308 \\   
&& 5 &   0.0000 &   0.0000 &   0.0000 &   0.0003 &   0.0010 &   0.0027 &   0.0056 &   0.0096 &   0.0148 \\ \cline{2-12}
& $5\%$ & 1 &   0.0046 &   0.0256 &   0.0580 &   0.0934 &   0.1267 &   0.1568 &   0.1832 &   0.2067 &   0.2271 \\    
&& 2 &   0.0001 &   0.0015 &   0.0077 &   0.0196 &   0.0357 &   0.0542 &   0.0734 &   0.0927 &   0.1113 \\    
&& 3 &   0.0000 &   0.0001 &   0.0010 &   0.0040 &   0.0098 &   0.0183 &   0.0288 &   0.0409 &   0.0539 \\   
&& 4 &   0.0000 &   0.0000 &   0.0001 &   0.0008 &   0.0026 &   0.0061 &   0.0111 &   0.0178 &   0.0257 \\    
&& 5 &   0.0000 &   0.0000 &   0.0000 &   0.0002 &   0.0007 &   0.0020 &   0.0042 &   0.0076 &   0.0121 \\ \hline 
10 & $3\%$& 1 &   0.0112 &   0.0493 &   0.0978 &   0.1448 &   0.1856 &   0.2203 &   0.2494 &   0.2749 &   0.2964 \\   
&& 2 &   0.0003 &   0.0049 &   0.0190 &   0.0411 &   0.0669 &   0.0936 &   0.1193 &   0.1442 &   0.1669 \\    
&& 3 &   0.0000 &   0.0005 &   0.0035 &   0.0113 &   0.0236 &   0.0391 &   0.0564 &   0.0748 &   0.0932 \\    
&& 4 &    0.0000 &   0.0000 &   0.0006 &   0.0030 &   0.0081 &   0.0160 &   0.0262 &   0.0383 &   0.0515 \\   
&& 5 &    0.0000 &   0.0000 &   0.0001 &   0.0008 &   0.0027 &   0.0064 &   0.0119 &   0.0193 &   0.0281 \\\cline {2-12}
& $5\%$ & 1 &   0.0049 &   0.0282 &   0.0654 &   0.1064 &   0.1452 &   0.1800 &   0.2103 &   0.2372 &   0.2605 \\    
&& 2 &   0.0001 &   0.0020 &   0.0101 &   0.0256 &   0.0462 &   0.0695 &   0.0932 &   0.1168 &   0.1392 \\    
&& 3 &   0.0000 &   0.0001 &   0.0015 &   0.0061 &   0.0146 &   0.0267 &   0.0411 &   0.0574 &   0.0742 \\    
&& 4 & 0.0000 &   0.0000 &   0.0002 &   0.0014 &   0.0046 &   0.0101 &   0.0180 &   0.0280 &   0.0393 \\    
&& 5 &   0.0000 &   0.0000 &   0.0000 &   0.0003 &   0.0014 &   0.0038 &   0.0078 &   0.0135 &   0.0206   
\\\hline
\end{tabular}

\end{table}
}

From Table~\ref{tab:ruin_exp} we obtain the following conclusions.

\begin{itemize}
\item
If the  initial capital grows, the part of the insurer's retained loss also grows with the constant level of risk of the company bankruptcy for any time horizon $n$.
\item
If the initial invention rate grows then the level of retention $b$ also grows with the constant ruin probability for any time horizon $n$.
\item
If time horizon $n$ grows, then the ruin probability grows for every fixed $u\geq 0.2$ and interest rate $I_0=i_s$. The greater $u$, the smaller ruin probability.
\end{itemize}

{
\begin{table}[!hbt]
\centering
\caption{\label{tab:b0_exp}Maximal level of retention $b$, for which the ruin probability does not exceed $0.05$ for exponential distribution.}

\smallskip

\begin{tabular}{|*{7}{c|}}
\hline 
\multicolumn{2}{|c|}{Initial capital $u$}  & 1 & 2 & 3 & 4 & 5 \\
\hline\hline 
$n=5$ & $i_s=3\%$ & 0.3289 & 0.6188 & 0.9062 & 1.0000 & 1.0000 \\
      & $i_s=5\%$ & 0.3752 & 0.6775 & 0.9700 & 1.0000 & 1.0000 \\
\hline 
$n=10$ & $i_s=3\%$ & 0.3005 & 0.5339 & 0.7626 & 0.9876 & 1.0000 \\
       & $i_s=5\%$ & 0.3585 & 0.6160 & 0.8549 & 1.0000 & 1.0000 \\\hline
\end{tabular}

\end{table}
}

Table~\ref{tab:b0_exp} implies that with initial capital $u\geq 4$ and interest rate $I_0=i_s=0.03$ for every $b$ the ruin probability does not exceed $0.05$ for time horizon $n=5$ and $n=10$. This means that without using an insurance the insurer is exposed to bankruptcy with a small probability not exceeding $5\%$.

In Table \ref{tab:b0_exp} the number 1 means that without reinsurance an insurer will have the level of bankruptcy below $5\%$.

We calculate the parameter $\xi\(b\)$ from Equation \eqref{eq:sup-beta} for $V\(x\)$ defined by \eqref{eq:expo_pdf}:
\begin{equation}\label{eq:sup_beta_expo}
\xi\(b\)=\sup_{x\geq c\(b\)}\frac{e^{R\(b\)x}\bar{V}\(x\)}
{\ds\int_x^{\infty}e^{R\(b\)z}dV\(z\)}
=\sup_{x\geq c\(b\)}\frac{e^{R\(b\)x}e^{-x/b}}
{\ds\int_x^{\infty}e^{R\(b\)z}\frac{1}{b}e^{-z/b}\,dz}\,.
\end{equation}
We calculate the integral under assumption that $bR\(b\)<1$:
\[
\int\limits_x^{\infty}e^{R\(b\)z}\frac{1}{b}e^{-z/b}\,dz
=\frac{1}{R\(b\)-1/b}\,e^{\(R\(b\)-1/b\)z}\Big|_{z=x}^{z=\infty}
=\frac{1}{1-bR\(b\)}e^{\(R\(b\)-1/b\)x}.
\]
After substitution to \eqref{eq:sup_beta_expo} we have
\[
\xi\(b\)=\sup_{x\geq c\(b\)}\frac{e^{\(R\(b\)-1/b\)x}}
{\frac{1}{1-bR\(b\)}e^{\(R\(b\)-1/b\)x}}.
\]
Hence
\begin{equation}\label{eq:sup_beta_expo_res}
\xi\(b\)=1-bR\(b\).
\end{equation}
For the parameter $b>1-\theta/\eta$ the adjustment coefficient $R\(b\)$ is the positive solution of Equation~\eqref{eq:coeff_R}.  Since the moment generating function $V\(x\)$ has the form
\[
M\(z\)=\frac{1}{1-bz},
\]
where $z<1/b$, then Equation~\eqref{eq:coeff_R} has the form
\[
\frac{1}{1-bR\(b\)}=e^{R\(b\)c\(b\)}
\]
from which we determine $R\(b\)$.

Based on Theorem~\ref{thm:est-up1}, the upper estimation of the ruin probability has the form
\begin{equation}\label{eq:upp_ruin}
\Psi_n^b\(u,i_s\)\leq
\(1-bR\(b\)\)\sum_{t=1}^l p_{st} e^{-R\(b\)u\(1+i_t\)},
\quad
n=1,2,\dots
\end{equation}
Let us denote the right-hand-side of the inequality~\eqref{eq:upp_ruin} by $g^b\(u,i_s\)$.
Figure \ref{fig:ueU2} depictes graphs of $\Psi_n^b\(u,i_s\)$
for an exponential distribution for $n=5$ and $n=10$, for each one for $b=0.2,\,0.4,\dots,1.0$ and for $i_2=0.05$.
In Figure \ref{fig:ue2} graphs of $\Psi_n^b\(u,i_s\)$ for $n=5$ and $n=10$ were depicted, for $u=1,\,2,\,3,\,4,\,5$ and for $i_2=0.05$. Graphs for $i_1=0.03$ are almost the same  so we omit them. The differences are easy to observe in Table~\ref{tab:ruin_exp}.

\begin{figure}[!ht]
\centering
\includegraphics{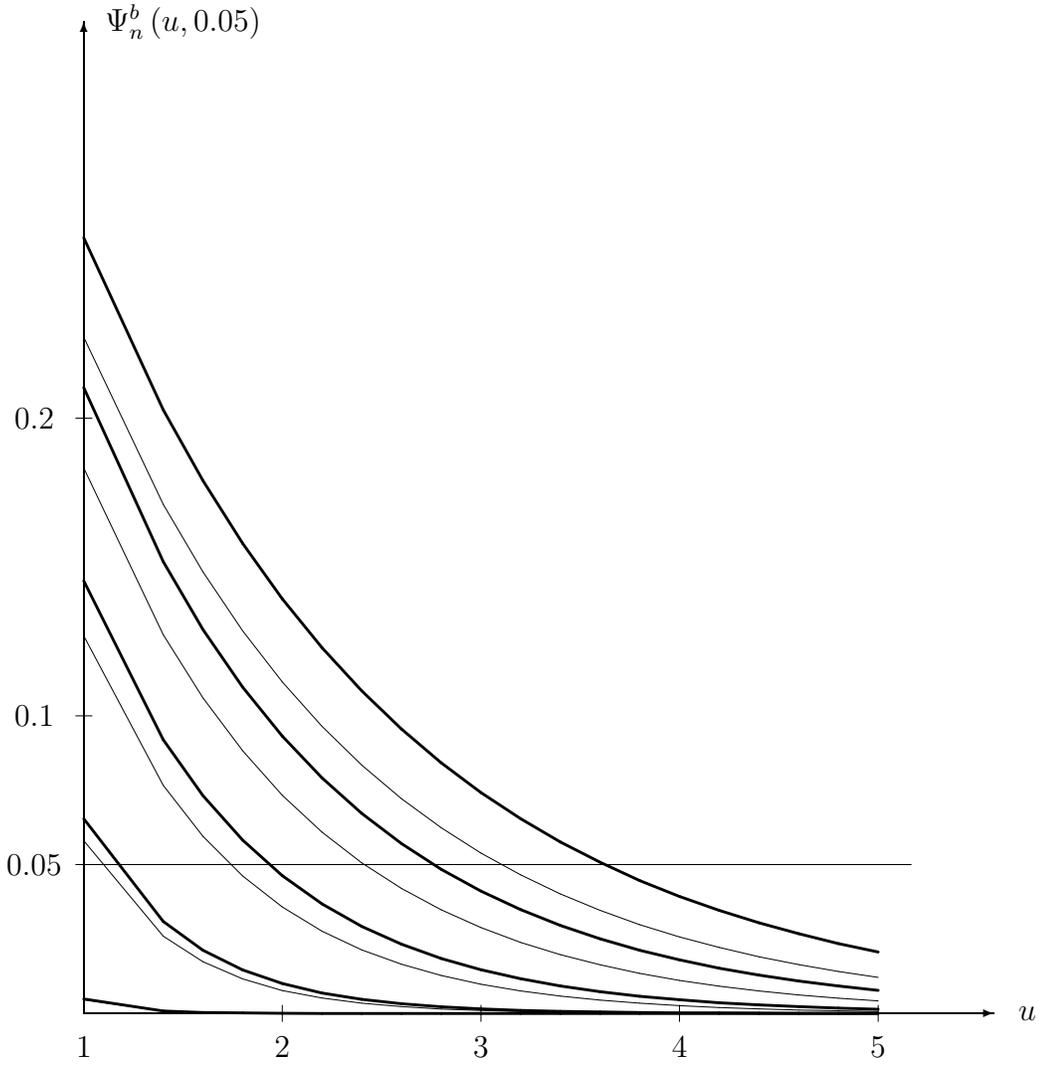}

\vspace{2ex}
\caption{\label{fig:ueU2}Ruin probability for exponential distribution as a function of~$u$. $\Psi_5^b\(u,0.05\)$ -- thin line, $\Psi_{10}^b\(u,0.05\)$ -- thick line, from the lowest to the highest for $b=0.2,\,0.4\,,0.6\,,0.8,\,1.0$ respectively.}
\end{figure}

\begin{figure}[!htb]
\centering
\includegraphics{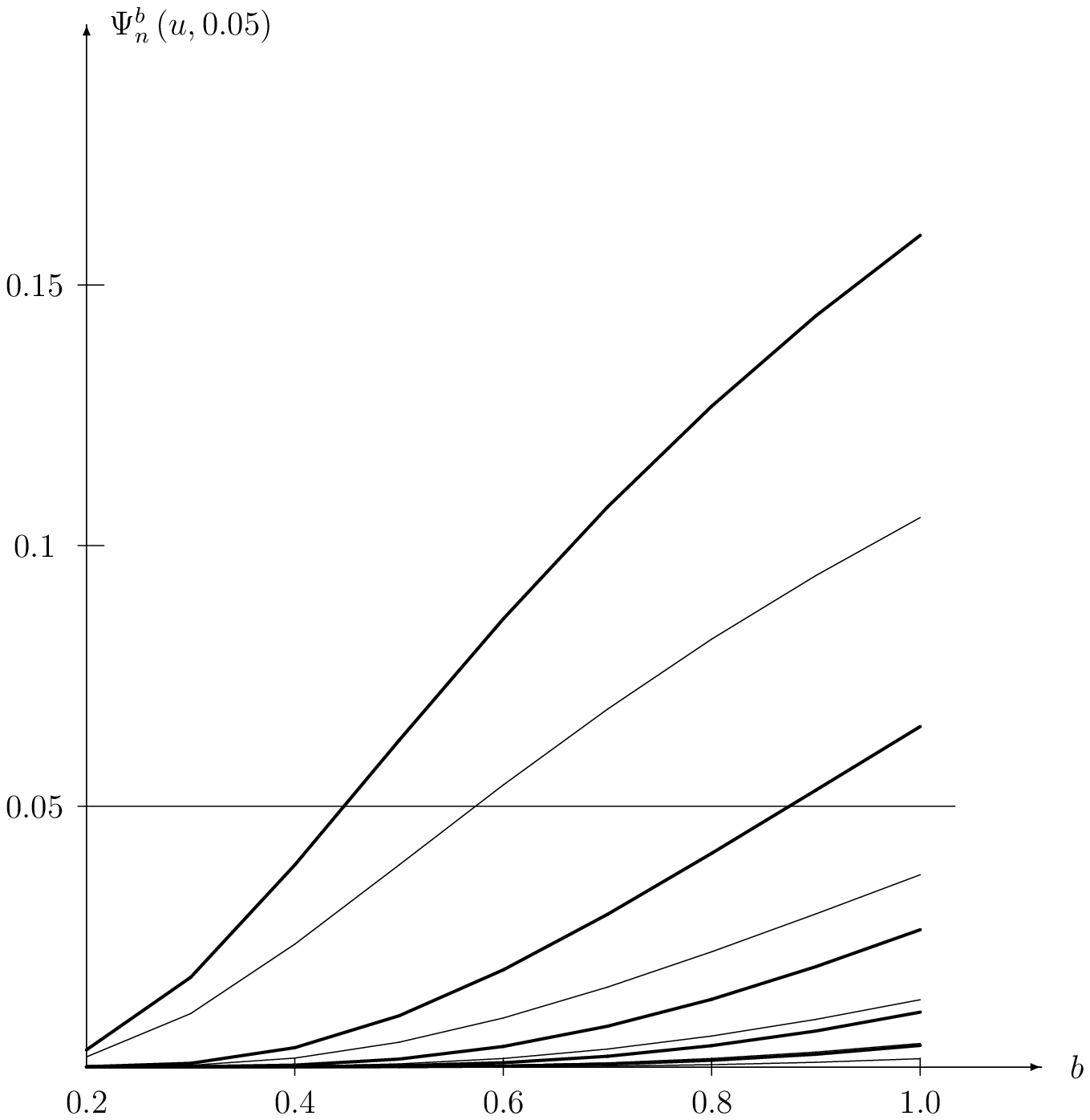}

\vspace{2ex}
\caption{\label{fig:ue2}Ruin probability for exponential distribution as a function of~$b$. $\Psi_5^b\(u,0.05\)$ -- thin line, $\Psi_{10}^b\(u,0.05\)$ -- thick line, from the highest do the lowest for $u=1,2,3,4,5$ respectively.}
\end{figure}

\subsection{Pareto distribution}
\label{ss:Pareto}

We assume that the total loss $Z_n$ has Pareto distribution with the distribution function
\begin{equation}\label{eq:Pareto_pdf}
W\(x\)=1-\(\frac{\beta}{x}\)^{\alpha}
\end{equation}
for $x\geq\beta>0$.
The random variable $Z_n$ has the expectation
\[
\E X=\frac{\alpha\beta}{\alpha-1}
\]
for $\alpha>1$ and a variance
\[
\V X=\frac{\alpha\beta^2}{\(\alpha-1\)^2\(\alpha-2\)}
\]
for $\alpha>2$.

We assume that $\E Z_n=1$. Hence the parameter $\beta$ must be in the form
\[
\beta=\frac{\alpha-1}{\alpha}\,.
\]
The loss $Z_n^{\ce}=bZ_n$ retained by insurer has cdf
\begin{equation}\label{eq:Pareto_strata}
V\(x\)=1-\(\frac{b\beta}{x}\)^{\alpha}
\end{equation}
for $x\geq b\beta$.

In the numerical calculations we assume $\alpha=1.25$ similarly to the paper by~Palmowski~\cite{Palmowski:Approx}. In this paper it was show that the greatest losses which came out at the end of eighties and nineties of XX century have Pareto distribution with the parameter approximately equal to  1.24138. With such a value of $\alpha$ the variance is infinite.

From \eqref{eq:c(b)} we have
\[
c\(b\)=\(1+\eta\)b-\(\eta-\theta\).
\]
The function $\Psi_1^b\(u,i_s\)$ can be set by \eqref{eq:ruin-recc11} in explicit form only for $n=1$, $l=1$ and $i_1=0$.
\begin{equation}
\Psi_1^b\(u\)=
\left( \frac{b\beta}{u+\theta+b\left(\eta+1\right) -\eta}\right)^{\alpha} \\
\end{equation}
The cases $n>1$ need numerical integrations. Let us consider the case $n=2$. In this case it is necessary to calculate the integral
\[
\alpha{\left( b\beta\right) }^{\alpha}\int\limits_{b\,\beta}^{x+c\(b\)}{\left( \frac{b\beta}{u+\theta+b\left( \eta+1\right) -\eta-z}\right) }^{\alpha}{z}^{-\(\alpha+1\)}dz.
\]
Substituting $A=u+\theta+b\left( \eta+1\right) -\eta$ we come to the problem of the calculation of the integral
\[
\int\limits\frac{1}{\(A-z\)^{\alpha}z^{\alpha+1}}\,dz
=-\frac{\(1-z/A\){}_2F_1\(-\alpha,\alpha;1;1-\alpha,x/A\)}
{\alpha\(A-z\)^{\alpha}z^{\alpha}}\,,
\]
where $_2F_1\(a,b;c;z\)$ is the hypergeometric function.

{
\tabcolsep 4pt
\begin{table}[!hbt]
\centering
\caption{\label{tab:ruin_Pareto}Values of ruin probabilities for Pareto distribution}

\smallskip

\footnotesize

\begin{tabular}{|c|c|c|*{9}{c}|}
\hline
$n$ & $i_s$ & $u$ & \multicolumn{9}{c|}{$b$} \\\cline{4-12} 
& & &  0.2 & 0.3 & 0.4 & 0.5 & 0.6 & 0.7 & 0.8 & 0.9 & 1.0 \\\hline
5  & $3\%$ & 1 & 0.0471 &   0.0663 &   0.0818 &   0.0945 &   0.1050 &   0.1156 &   0.1214 &   0.1280 &   0.1337 \\   
 & &2 & 0.0255 &   0.0384 &   0.0499 &   0.0602 &   0.0693 &   0.0787 &   0.0846 &   0.0912 &   0.0972 \\     
 & &3 & 0.0169 &   0.0263 &   0.0352 &   0.0434 &   0.0510 &   0.0590 &   0.0644 &   0.0704 &   0.0759 \\     
 & &4 & 0.0124 &   0.0197 &   0.0267 &   0.0335 &   0.0399 &   0.0468 &   0.0515 &   0.0569 &   0.0618 \\     
 & &5 & 0.0097 &   0.0156 &   0.0214 &   0.0270 &   0.0325 &   0.0384 &   0.0427 &   0.0474 &   0.0519 \\ \cline{2-12}
 &$5\%$ &1 & 0.0421 &   0.0603 &   0.0754 &   0.0880 &   0.0986 &   0.1092 &   0.1154 &   0.1222 &   0.1281 \\     
 & &2 & 0.0234 &   0.0356 &   0.0466 &   0.0566 &   0.0655 &   0.0748 &   0.0807 &   0.0873 &   0.0933 \\     
 & &3 & 0.0158 &   0.0247 &   0.0332 &   0.0411 &   0.0485 &   0.0563 &   0.0617 &   0.0676 &   0.0730 \\     
 & &4 & 0.0118 &   0.0187 &   0.0254 &   0.0320 &   0.0382 &   0.0448 &   0.0495 &   0.0548 &   0.0597 \\     
 & &5 & 0.0092 &   0.0148 &   0.0204 &   0.0259 &   0.0312 &   0.0370 &   0.0411 &   0.0458 &   0.0502 \\ \hline
10 &$3\%$ &1 & 0.0685 &   0.0947 &   0.1150 &   0.1312 &   0.1442 &   0.1599 &   0.1640 &   0.1718 &   0.1785 \\   
 & &2 & 0.0405 &   0.0599 &   0.0765 &   0.0907 &   0.1029 &   0.1175 &   0.1226 &   0.1309 &   0.1382 \\   
 & &3 & 0.0282 &   0.0432 &   0.0568 &   0.0690 &   0.0798 &   0.0929 &   0.0980 &   0.1060 &   0.1131 \\   
 & &4 & 0.0213 &   0.0334 &   0.0448 &   0.0552 &   0.0648 &   0.0765 &   0.0814 &   0.0888 &   0.0956 \\   
 & &5 & 0.0170 &   0.0270 &   0.0367 &   0.0458 &   0.0542 &   0.0647 &   0.0694 &   0.0762 &   0.0826 \\ \cline{2-12} 
 &$5\%$ &1 & 0.0582 &   0.0829 &   0.1028 &   0.1191 &   0.1325 &   0.1480 &   0.1532 &   0.1615 &   0.1687 \\   
 & &2 & 0.0354 &   0.0533 &   0.0691 &   0.0829 &   0.0949 &   0.1091 &   0.1148 &   0.1232 &   0.1307 \\   
 & &3 & 0.0252 &   0.0391 &   0.0519 &   0.0635 &   0.0740 &   0.0866 &   0.0921 &   0.1000 &   0.1072 \\   
 & &4 & 0.0194 &   0.0306 &   0.0413 &   0.0513 &   0.0605 &   0.0717 &   0.0768 &   0.0841 &   0.0908 \\   
 & &5 & 0.0156 &   0.0250 &   0.0341 &   0.0427 &   0.0509 &   0.0609 &   0.0656 &   0.0723 &   0.0786 \\ \hline 
\end{tabular}

\end{table}
}

{
\begin{table}[!hbt]
\centering
\caption{\label{tab:b0_Pareto}Maximal level of retention $b$, for which the ruin probability does not exceed $0.05$ for Pareto distribution.}

\smallskip

\begin{tabular}{|*{7}{c|}}
\hline 
\multicolumn{2}{|c|}{Initial capital $u$} & 1 & 2 & 3 & 4 & 5 \\
\hline\hline 
$n=5$ & $i_s=3\%$ & 0.2190 & 0.4052 & 0.5907 & 0.7696 & 0.9621 \\
      & $i_s=5\%$ & 0.2468 & 0.4379 & 0.6209 & 0.8133 & 0.9996  \\\hline
$n=10$ & $i_s=3\%$ & lack  & 0.2567 & 0.3582 & 0.4588 & 0.5588  \\
       & $i_s=5\%$ & lack & 0.2884 & 0.3933 & 0.4958 & 0.5974 \\\hline
\end{tabular}

\end{table}
}

Table~\ref{tab:ruin_Pareto} gives the same conclusion as for exponential distribution.
Word ``lack'' in Table~\ref{tab:b0_Pareto} means that for any level of retention $b\in\(0.2,1\right]$ with initial capital $u=1$, the ruin probability exceeds $0.05$ both for a five-years-time horizon and for a ten-year-time horizon.
In Figure~\ref{fig:upU2} graphs of $\Psi_n^b\(u,i_s\)$ for $n=5$ and $n=10$ for Pareto distribution were depicted for  $i_2=0.05$. In Figure~\ref{fig:up2} graphs of $\Psi_n^b\(u,i_s\)$ for $n=5$ and $n=10$, for $u=1,\,2,\,3,\,4,\,5$ and $i_2=0.05$. Graphs for $i_1=0.03$ are almost the same so we omit them. The differences are easy to observe in Table~\ref{tab:ruin_Pareto}.

\begin{figure}[!ht]
\centering
\includegraphics{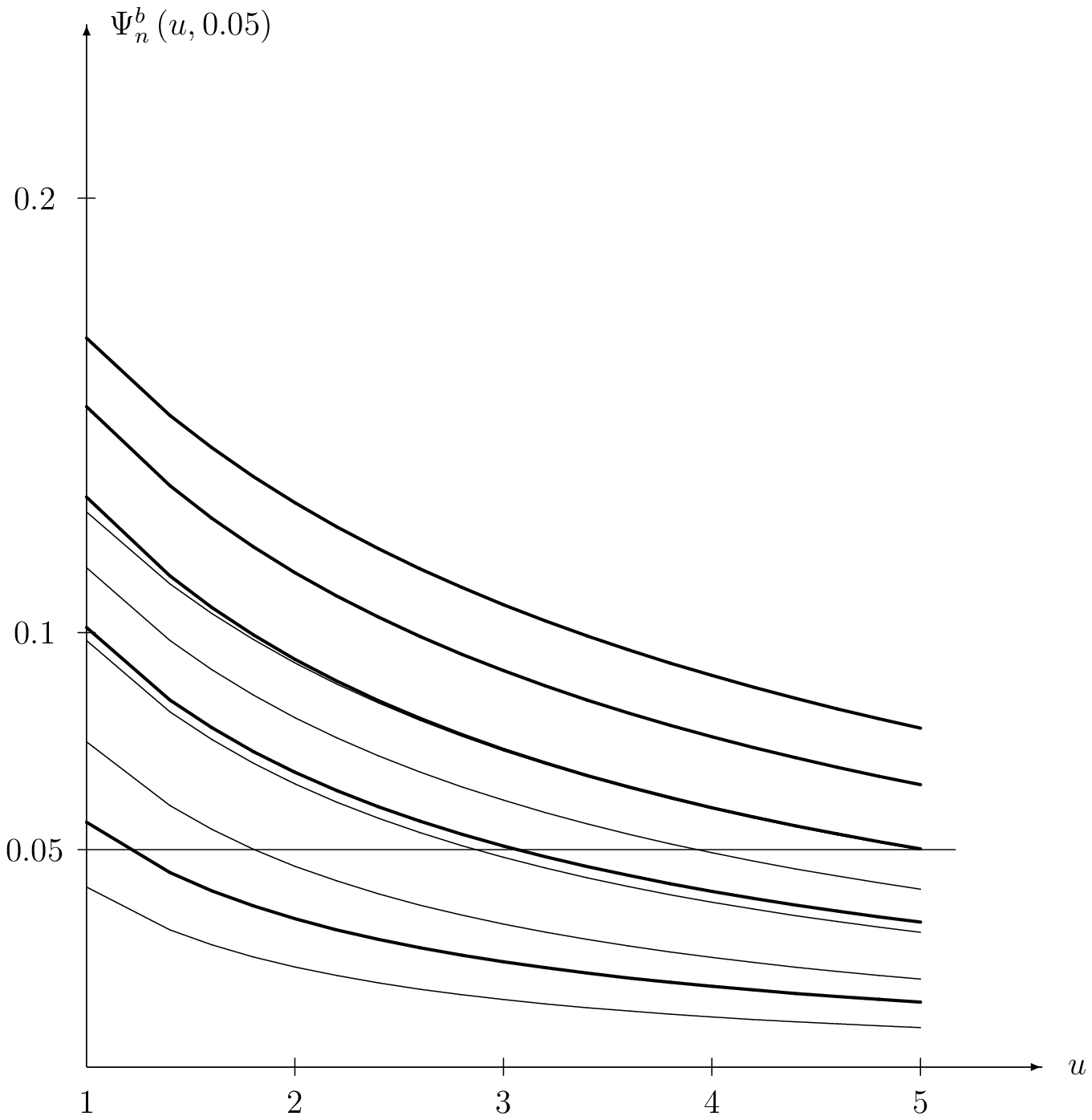}

\vspace{2ex}
\caption{\label{fig:upU2}Ruin probability for Pareto distribution as a function of~$u$. $\Psi_5^b\(u,0.05\)$ -- thin line, $\Psi_{10}^b\(u,0.05\)$ -- thick line, from the lowest do the highest for $b=0.2,\,0.4\,,0.6\,,0.8,\,1.0$ respectively.}
\end{figure}

\begin{figure}[!htb]
\centering
\includegraphics{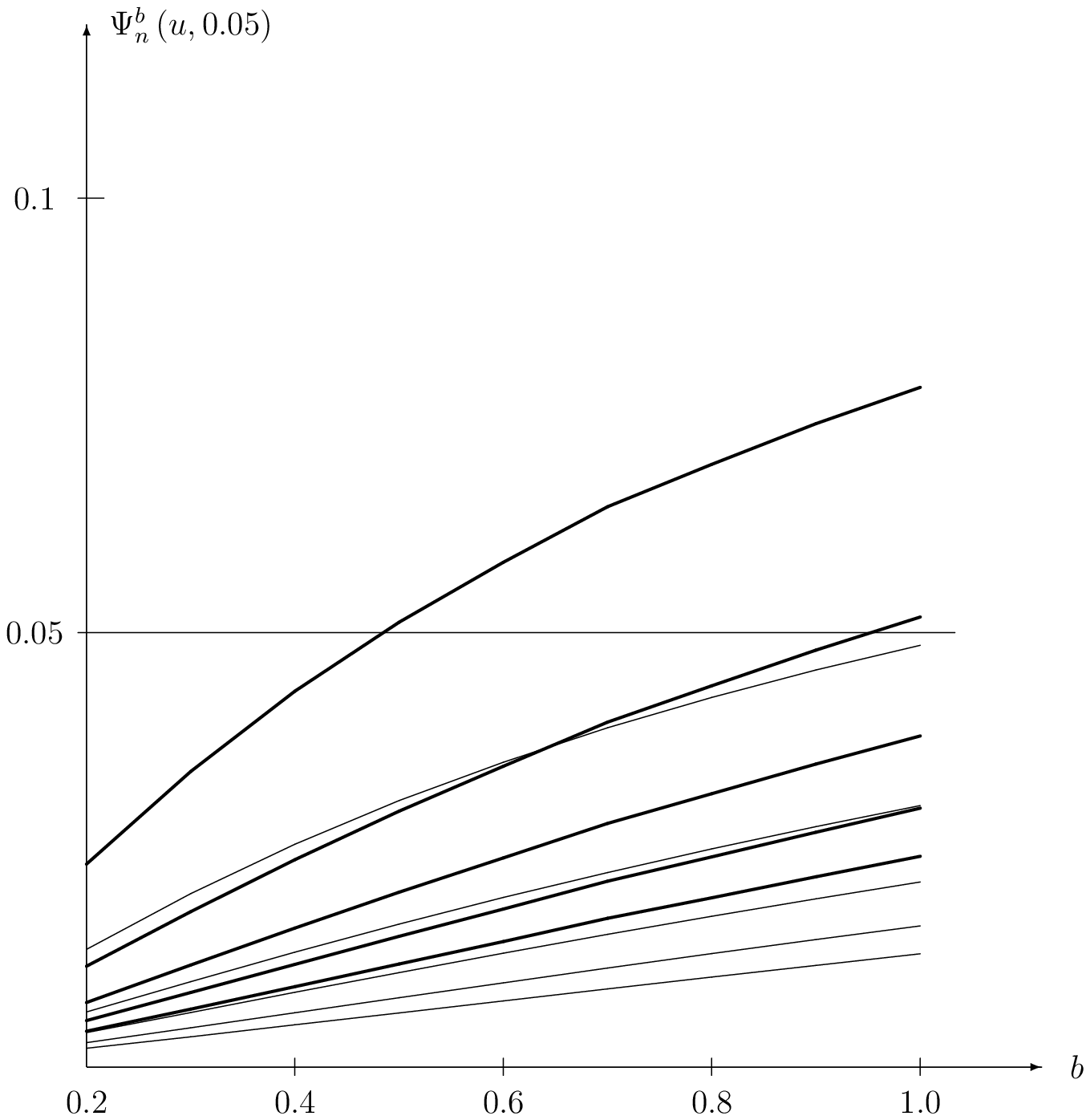}

\vspace{2ex}
\caption{\label{fig:up2}Ruin probability for Pareto distribution as a function of~$b$. $\Psi_5^b\(u,0.05\)$ -- thin line, $\Psi_{10}^b\(u,0.05\)$ -- thick line, from the highest do the lowest for $1,2,3,4,5$ respectively.}
\end{figure}

Taking an advantage from Theorem~\ref{thm:est-up2} we will present the results concerning an approximation of ruin probability for Pareto distribution. In Figure~\ref{fig:apro} the ratio
\[
\frac{\Psi_n^b\(u,i_s\)}{c_n\(i_s\)\bar{V}\(u\)}
\]
for~$n=3$, $b=0.2,0.4,\dots,1.0$ and $0\leq u\leq 20$ was depicted.

\begin{figure}[!ht]
\centering
\includegraphics[width=16cm]{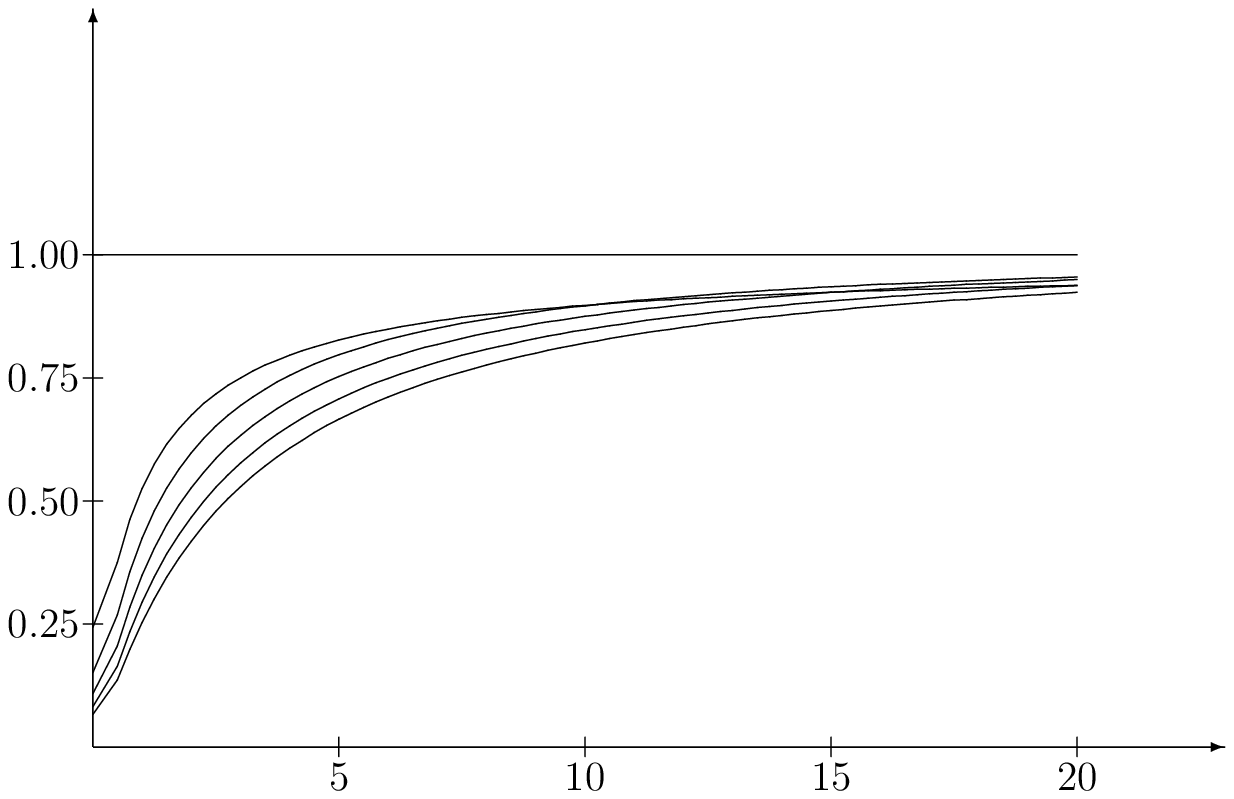}
\caption{\label{fig:apro}Asymptotic approximation of the ruin probability for Pareto distribution --
graphs
${\Psi_n^b\(u,i_s\)}/{c_n\(i_s\)\bar{V}\(u\)}$.
From the highest do the lowest for $b=0.2,\,0.4,\,0.6,\,0.8,\,1.0$ respectively}

\bigskip
\end{figure}

\section{Conclusions}
\label{s:concl}

In the continuous risk process the optimal level of retention can be determined by maximising of an adjustment coefficient relative to the level of retention.
In the discrete risk process the above statement is not true.

For the fixed initial capital $u\geq 1$ the probability of ruin is an increasing function of the retention level $b$. Therefore the probability of the ruin is minimal if the retention level is minimal. It means that an insurer retains only very low losses which causes very low income and is very unfavourable for him. It seems that the right approach relies on fixing an acceptable level of the ruin probability, and appropriately to this probability, determining the retention level.

If loading of a reinsurer is greater than loading of an insurer ($\xi>\theta$), the adjustment coefficient is not a convex function, which lowers the quality of upper estimation.
Basing on our numerical examples we conclude that such an upper bound is very imprecise, and basically it is worthless.
For the heavy tailed claims we give the theorem about the approximation of the ruin probability if the initial capital is sufficiently large. The example of Pareto distribution shows that such an approximation is appropriate and quickly tends to the limit value.

\newpage
\section*{Acknowledgements}
The research by Helena Jasiulewicz was supported by a grant from the National Science Centre, Poland.


\begin{thebibliography}{10}

\bibitem{Cai:Discrete}
CAI, J., Discrete time risk models under rates of interest.
\newblock Prob. Eng. Inf. Sci. \textbf{16}, 309--324 (2002)

\bibitem{Cai:Ruin}
CAI, J., Ruin probabilities with dependent rates of interest.
\newblock J. Appl. Prob. \textbf{39}, 312--323 (2002)

\bibitem{CaiDickson:ruin}
CAI, J., DICKSON, D.C.M., Ruin probabilities with a {M}arkov chain interest
  model.
\newblock Insurance Math. Econom. \textbf{35}, 513--525 (2004)

\bibitem{DiasparaRomera:Bounds}
DIASPARRA, M.A., ROMERA, R., Bounds for the the ruin probability of a
  discrete-time risk process.
\newblock J. Appl. Probab. \textbf{46}, 99--112 (2009)

\bibitem{DicksonWaters:Reinsurance}
DICKSON, D.C.M., WATERS, H.R., Reinsurance and ruin.
\newblock Insurance Math. Econom. \textbf{19}, 61--80 (1996)

\bibitem{HJ:surplus}
JASIULEWICZ, H., Discrete-time financial surplus models for insurance
  companies.
\newblock Annals of the Collegium of Economic Analysis \textbf{21}, 225--255
  (2010)

\bibitem{HJ:ruin-discr}
JASIULEWICZ, H., Discrete risk process with reinsurance and random interest
  rate.
\newblock Annals of the Collegium of Economic Analysis \textbf{31}, 11--26
  (2013).
\newblock (in polish)

\bibitem{Palmowski:Approx}
PALMOWSKI, Z., Approximations of ruin probability of insurance company in
  diffusion {C}ox model.
\newblock Research Papers of Wroc{\l}aw University of Economics \textbf{1108},
  34--64 (2006).
\newblock (in polish)

\bibitem{TangTsit_:Precise}
TANG, Q., TSITSIASHVILI, G., Precise estimates for the ruin probability in
  finite horizon in a discrete-time model with heavy-tailed insurance and
  financial risk.
\newblock Stochastic Processes Appl. \textbf{108}, 299--325 (2003)

\bibitem{Yang:Non-exp}
YANG, H., Non-exponential bounds for ruin probability with interest effect
  included.
\newblock Scand. Actuarial J. \textbf{99}, 66--79 (1999)

\end{thebibliography}
\end{document}